\documentclass[journal]{IEEEtran}
\usepackage{amsmath,amsfonts}
\usepackage{algorithmic}
\usepackage{algorithm}
\usepackage{array}
\usepackage[caption=false,font=normalsize,labelfont=sf,textfont=sf]{subfig}
\usepackage{textcomp}
\usepackage{stfloats}
\usepackage{url}
\usepackage{verbatim}
\usepackage{graphicx}
\usepackage{amsmath}
\usepackage{amssymb}
\usepackage{mathrsfs}
\usepackage{makecell}
\usepackage{threeparttable}
\usepackage[figuresright]{rotating}
\usepackage{soul} 
\usepackage{color, xcolor} 
\usepackage{booktabs}

\newtheorem{lemma}{Lemma}
\newtheorem{corollary}{Corollary}
\newtheorem{theorem}{Theorem}
\newtheorem{definition}{Definition}
\newenvironment{proof}{{\noindent\it Proof:} }{\hfill $\square$\par}

\newtheorem{example}{Example}
\newtheorem{remark}{Remark}
\begin{document}
	
	\title{Some quaternary additive codes outperform linear counterparts}
	
	\author{Chaofeng Guan, Ruihu Li, Yiting Liu, Zhi Ma~\IEEEmembership{}
	}
	\markboth{}%
	{}
	
	\IEEEpubid{}
	
	\maketitle
	
	\begin{abstract}
		The additive codes may have better parameters than linear codes. However, it is still a challenging problem to efficiently construct additive codes that outperform linear codes, especially those with greater distances than linear codes of the same lengths and dimensions.
		This paper focuses on constructing additive codes that outperform linear codes based on quasi-cyclic codes and combinatorial methods.
		Firstly, we propose a lower bound on the symplectic distance of 1-generator quasi-cyclic codes of index even.  
		Secondly, we get many binary quasi-cyclic codes with large symplectic distances utilizing computer-supported combination and search methods, all of which correspond to good quaternary additive codes. Notably, some additive codes have greater distances than best-known quaternary linear codes in Grassl's code table (bounds on the minimum distance of quaternary linear codes http://www.codetables.de) for the same lengths and dimensions. Moreover, employing a combinatorial approach, we partially determine the parameters of optimal quaternary additive 3.5-dimensional codes with lengths from $28$ to $254$.
		Finally, as an extension, we also construct some good additive complementary dual codes with larger distances than the best-known quaternary linear complementary dual codes in the literature.
	\end{abstract}
	

	\section{Introduction}
	\IEEEPARstart{O}{ne} of the most significant problems in coding theory is constructing good error-correcting codes. 
	After decades of efforts, scholars have constructed a large number of linear codes with suitable parameters, Grassl et al. summarized those results and established an online code table \cite{Grassltable} of best-known linear codes over small finite fields $\mathbb{F}_q$, $q\le 9$.
	Unlike linear codes, additive codes are closed under vector addition but not necessarily closed under scalar multiplication. All linear codes can be considered as also additive codes, but additive codes are not necessarily linear. Therefore, theoretically, additive codes may have better parameters than linear codes. In addition, additive codes also have critical applications in quantum information \cite{calderbank1998quantum,ketkar2006nonbinary}, computer memory systems \cite{chen1984error,chen1991fault,chen1992symbol}, deep space communication \cite{hattori1998subspace}, and secret sharing \cite{kim2017secret}. Thus, it is crucial to construct good additive codes, especially ones with better performance than best linear codes.

	Quaternary additive codes were the first to receive scholarly attention, given the links to communications, electronic devices, computers, etc.
	In \cite{blokhuis2004small}, Blokhuis and Brouwer determined the parameters of optimal quaternary additive codes of lengths not more than 12, some of which have higher information rates than optimal linear cases.
	Afterward, much work has been done on quaternary additive codes with small lengths \cite{bierbrauer2009short,bierbrauer2010geometric,bierbrauer2015nonexistence,bierbrauer2019additive} or low dimensions \cite{guo2017construction,bierbrauer2021optimal}, resulting in a general determination of the parameters of quaternary additive codes of lengths up to 15 and a complete determination of the parameters of 2.5-dimensional optimal quaternary additive codes. Meanwhile, additive complementary dual codes\footnote{In this paper, all additive complementary dual codes are with respect to the trace Hermitian inner product, which can be simplified to Hermitian inner product, in linear case \cite{calderbank1998quantum}, and all linear complementary dual codes (LCD codes) are with respect to the Hermitian inner product.} 
	(ACD codes) have also attracted a wild attention of scholars owing to their utility in constructing maximal-entanglement entanglement-assisted quantum codes \cite{lai2017linear,xu2021constructions,huang2022constructions}, and their application in resisting side-channel attacks \cite{carlet2016complementary,shi2023additive,benbelkacem2020z2z4,shi2022additive,dougherty2022additive}.
	
	Quasi-cyclic codes are an interesting class of linear codes that exhibit good performance in constructing record-breaking or best-known linear codes in \cite{Grassltable}. 
	With appropriate mappings, quasi-cyclic code can be used to construct good additive codes \cite{galindo2018quasi,guneri2018additive,shi2018asymptotically,shi2021Z2Z4,guan2022symplectic,Guan2023OnEH}.
	However, there is still a lack of practical approaches to construct additive codes using quasi-cyclic codes, which makes it challenging to construct good additive codes with quasi-cyclic codes.
	This paper proposes a lower bound on the symplectic distance of 1-generator quasi-cyclic codes of index even and several combinatorial construction methods of additive codes. 
	Further, we construct many good quaternary additive codes and ACD codes, which perform better than linear counterparts in Grassl's code table \cite{Grassltable} or best-known LCD codes in \cite{lu2020optimal,harada2021construction,Ishizuka2022,ishizuka2022construction}.

	This paper is structured as follows.
	In the next section, we give some of the foundations used in this paper. In Sec. \ref{III}, we propose a lower bound on the symplectic distance of 1-generator quasi-cyclic codes of index even. In Sec. \ref{IV}, 
	we give some combinations, enhancements and derivations of additive codes, and construct a large number of good additive codes.
	In Sec. \ref{V}, we discuss the construction of ACD codes, and obtain some quaternary ACD codes that are better than the quaternary LCD codes in the literature.
	Finally, in Sec. \ref{VI}, we discuss our main results and future research directions.
	The parameters of linear or additive codes in this paper are computed by algebra software Magma \cite{bosma1997magma}.
	\section{Preliminaries}\label{II}
	This section presents some fundamentals of additive codes and quasi-cyclic codes. 
	For more details, refer to  \cite{huffman2010fundamentals,bierbrauer2017introduction,huffman2021concise}. 
	
	\subsection{Additive codes}
	
	Let $p$ be a prime, and $\mathbb{F}_q$ is the finite field of order $q$, where $q = p^r$ for some positive integer
	$r$.
	For $\vec{u }=(u_{0},\ldots, u_{n-1})\in \mathbb{F}_{q}^{n}$, the Hamming weight of $\vec{u}$ is 
	$\mathrm{w}_{H}(\vec{u})=\#\left\{i \mid u_{i} \neq0, 0 \leq i \leq n-1\right\}$. Let $\vec{u }_1,\vec{u }_{2}\in \mathbb{F}_{q}^{n}$, then Euclidean inner product of them is $\langle\vec{u}_1, \vec{u}_{2}\rangle_{e}=\sum_{i=0}^{n-1} u_{1,i} u_{2,i}$.
	For $\vec{v}=(v_{0},\ldots, v_{2n-1}) \in \mathbb{F}_{q}^{2n}$,  
	symplectic weight of $\vec{v}$ is 
	$\mathrm{w}_{s}(\vec{v})=\#\left\{i \mid (v_{i}, v_{n+i}) \neq(0,0), 0 \leq i \leq n-1 \right\}  $.
	Let $\vec{v}_1,\vec{v}_{2}\in \mathbb{F}_{q}^{2n}$, then the symplectic inner product of them is  $\langle\vec{v}_1,\vec{v}_{2}\rangle_{s}=\sum_{i=0}^{n-1}\left(v_{1,i} v_{2,n+i}-v_{1,n+i} v_{2,i}\right)$.

	A code $\mathscr{C}_l$ is said to be linear over $\mathbb{F}_q$ if it is a linear subspace of $\mathbb{F}_q^n$.  
	If $\mathscr{C}_l$ have dimension $k$, minimum Hamming distance (weight) $d_H$, then $\mathscr{C}_l$ can be denoted as $[n,k,d_H]_q$.
	Linear codes of even lengths can also be considered as symplectic codes $\mathscr{C}_s$.
	If $\mathscr{C}_s$ is a $[2n,k]_q$ symplectic code of the minimum symplectic distance $d_s$, then $\mathscr{C}_s$ can be denoted as $[2n,k,d_s]^s_{q}$. 
	A code $\mathscr{C}_a$ is said to be an additive code over $\mathbb{F}_{q^2}$ if it is a subgroup of $\mathbb{F}^n_{q^2}$, this means that scalar multiples of the codewords do not necessarily belong to the code. Defining $k_a$ as the dimension of $\mathscr{C}_a$ over $\mathbb{F}_{q}$,
	then $\mathscr{C}_a$ have dimension $\frac{k_a}{2}$ over $\mathbb{F}_{q^2}$. If $\mathscr{C}_a$ have minimum Hamming distance $d_H$, then $\mathscr{C}_a$ can be denoted as $(n,\frac{k_a}{2},d_H)_{q^2}$.

	The Euclidean dual code of $\mathscr{C}_l$ is $\mathscr{C}_l^{\perp_{e}}=\left\{\vec{c}_1 \in \mathbb{F}_{q}^{n} \mid\langle\vec{c}_1, \vec{c}_{2}\rangle_{e}=0, \forall \vec{c}_{2} \in \mathscr{C}_l\right\}$.
	Symplectic dual of $\mathscr{C}_s$ is
	$ \mathscr{C}_s^{\perp_{s}}= \left\{\vec{c}_1 \in \mathbb{F}_{q}^{2n} \mid\langle\vec{c}_1, \vec{c}_{2}\rangle_{s}=0, \forall \vec{c}_{2} \in \mathscr{C}_s\right\}$.
	$\mathscr{C}_l$ is an Euclidean LCD code if and only if $\mathscr{C}_l \cap \mathscr{C}_l^{\perp_{e}} =\{\mathbf{0}\}$. $\mathscr{C}_s$ is a symplectic LCD code if and only if $\mathscr{C}_s \cap \mathscr{C}_s^{\perp_{s}} =\{\mathbf{0}\}$.

	To establish the connection between $\mathbb{F}^{2n}_q$ and $\mathbb{F}^{n}_{q^2}$. We define two maps, $\phi$ and $\Phi$, as follows. For $x,y\in\mathbb{F}_q$, then
	$\phi(x,y)=x+wy$, where $w$ generates $\mathbb{F}_{q^2}$ over $\mathbb{F}_{q}$.
	For $\vec{v}=(v_{0},\ldots, v_{2n-1}) \in \mathbb{F}_{q}^{2n}$,  $\Phi(\vec{v})=(\phi(v_0,v_n),\phi(v_1,v_{n+1}),\cdots,\phi(v_{n-1},v_{2n-1}))$.
	Clearly, mapping $\Phi$ is a one-to-one mapping, and there is $\mathrm{w}_H(\Phi(\vec{v}))=\mathrm{w}_s(\vec{v})$.
	
	Let $G_s$ denote generator matrix of $\mathscr{C}_s$, then $\Phi(G_s)$ can generate additive code $\mathscr{C}_a$ have parameters $(n,\frac{k}{2},d_s)_{q^2}$. Similarly, $\Phi^{-1}(\mathscr{C}_a)=\mathscr{C}_s$.  Therefore, a symplectic code $\mathscr{C}_s$ with parameters $[2n,k,d_s]^s_{q}$ is equivalent to an additive code $(n,\frac{k}{2} ,d_s)_{q^2}$. 
	Notably, Calderbank et al. \cite{calderbank1998quantum} also proved that binary symplectic inner product and quaternary trace Hermitian inner product are equivalent.
	Therefore, a binary symplectic LCD $[2n,k,d_s]_{2}^s$ code is also a quaternary ACD
	code with parameters $(n,\frac{k}{2} ,d_s)_{4}$. 
	

	In general, for an additive $(n,\frac{k}{2} ,d_s)_{q^2}$ code, where $k$ is even, if 
	the best linear code is $[n,  \frac{k}{2} ,< d_s]_{q^2}$, then we regard $(n,\frac{k}{2} ,d_s)_{q^2}$ as better than linear counterpart in a strong sense. When $k$ is odd, for specific $n$ and $d_s$, if the best linear code is $[n, \lfloor \frac{k}{2}\rfloor ,d_s]_{q^2}$, then $(n,\frac{k}{2} ,d_s)_{q^2}$ is regarded better for the higher information rate.  
	However, there is also a particular case in which the best linear codes are $[n, \lfloor \frac{k}{2}\rfloor,d_1]_{q^2}$ and $[n, \lceil \frac{k}{2}\rceil,d_2]_{q^2}$, with $d_2 < d_s < d_1$. In this case, we can consider that $(n,\frac{k}{2},d_s)_{q^2}$ fills the distance gap of best linear codes.

	\subsection{Cyclic codes and quasi-cyclic codes}
	
	Let $\mathscr{C}$ be an $[n,k]_q$ code over $\mathbb{F}_q$, $\mathscr{C}$  is cyclic provided that for all  $\mathbf{c}=(c_{0}, c_{1}, \cdots, c_{n-1}) \in \mathscr{C}$, the cyclic shift  $\mathbf{c}^{\prime}=(c_{n-1}, c_{0}, \cdots, c_{n-2}) \in \mathscr{C}$.
	Considering each codeword $\mathbf{c}$ as a coefficients vector of polynomial  
	$\mathbf{c}(x)= \sum_{i=0}^{n-1} c_{i} x^{i-1}$ in $\mathbb{F}_{q}[x]$, 
	then $\mathscr{C}$ corresponds to a principal ideal in the quotient ring  
	$\mathbb{R}_{q,n}=\mathbb{F}_{q}[x] /\left\langle x^{n}-1\right\rangle$, 
	which is generated by a unique monic non-zero polynomial $g(x)$ of degree $n-k$. 
	We call $g(x)$ the generator polynomial of cyclic code $\mathscr{C}$, and $\mathscr{C}$ also can be denoted as $\langle g(x) \rangle$. 
	The parity check polynomial of $\mathscr{C}$ is $h(x)=\left(x^{n}-1\right) / g(x)$. 
	The Euclidean dual code $\mathscr{C}^{\perp_e}$ of $\mathscr{C}$ is also a cyclic code with  generator polynomial
	$g^{\perp_e}(x)=\tilde{h}(x)=x^{\deg(h(x))} h(x^{-1})$.

	

	A linear code $\mathscr{C}$ of length $n\ell$ over $\mathbb{F}_{q}$ is called a quasi-cyclic code of index $\ell$ if  $\mathbf{c}=\left(c_{0}, c_{1}, \ldots, c_{n\ell-1}\right)$  is a codeword of $\mathscr{C}$,  then $\mathbf{c}^{\prime}=(c_{n-1},c_0, \ldots, c_{n-2}, c_{2n-1}, c_{n}, \ldots, c_{2n-2}, \ldots, c_{n\ell-1}, c_{(n-1)\ell},$ $ \ldots,  c_{n\ell-2})$ is also a codeword. 
	Circulant matrices are basic components in the generator matrix for quasi-cyclic codes. An $n\times n$ circulant matrix $M$ is defined as
	\begin{equation}
		M=\left(\begin{array}{ccccc}
			m_{0} & m_{1} & m_{2} & \ldots & m_{n-1} \\
			m_{n-1} & m_{0} & m_{1} & \ldots & m_{n-2} \\
			\vdots & \vdots & \vdots & \vdots & \vdots \\
			m_{1} & m_{2} & m_{3} & \ldots & m_{0}
		\end{array}\right).
	\end{equation}
	
	If the first row of  $M$  is mapped onto polynomial  $m(x)$, then circulant matrix  $M$  is isomorphic to polynomial $m(x)=m_{0}+m_{1} x+\cdots+m_{n-1} x^{n-1} \in \mathbb{R}_{q,n}$. So $M$ can be determined by polynomial $m(x)$. 
	Generator matrix of $h$-generator quasi-cyclic code with index $\ell$ has the following form:
	
	\begin{equation}
		G=\left(\begin{array}{cccc}
			M_{1,0} & M_{1,1} & \cdots & M_{1, \ell-1} \\
			M_{2,0} & M_{2,1} & \cdots & M_{2, \ell-1} \\
			\vdots & \vdots & \ddots & \vdots \\
			M_{h, 0} & M_{h, 1} & \cdots & M_{h, \ell-1}
		\end{array}\right),
	\end{equation}
	where  $M_{i, j}$  are circulant matrices generated by the polynomials  $m_{i, j}(x)\in \mathbb{R}_{q,n}$, where  $1 \leq i \leq h$  and  $0 \leq j \leq \ell-1$.

	\section{Bound on the symplectic weight of 1-generator quasi-cyclic codes of index even}\label{III}
	For narrative convenience, we fix $\ell$ as an even number and $m=\ell/2$, throughout this paper.
	Let $[s, t]$ ($s\le t$) denote the set $\{s,s+1,\cdots,t\}$.
	For $g(x)=g_{0}+g_{1} x+g_{2} x+\cdots+ g_{n-1} x^{n-1} \in \mathbb{R}_{q,n}$, $[g(x)]$ denote the vector generated by coefficients of $g(x)$ in $\mathbb{F}_{q}^{n}$, i.e., $[g(x)]=[{{g}_{0}},{{g}_{1}},{{g}_{2}},\cdots ,{{g}_{n-1}}]$. 
	
	First, we introduce the relationship between the symplectic and Hamming weights in Lemma \ref{symplectic_Hamming}.
	
	\begin{lemma}\label{symplectic_Hamming}(\cite{ling2010generalization})
		If  $\vec{x}$, $\vec{y}$ be two vectors in $\mathbb{F}_{q}^{n}$, then there is
		\begin{equation}
			q\cdot\mathrm{w}_{s}(\vec{x} \mid \vec{y})
			=\mathrm{w}_{H}(\vec{x})+\mathrm{w}_{H}(\vec{y})+\sum_{\alpha \in \mathbb{F}_{q}^{*}} \mathrm{w}_{H}(\vec{x}+\alpha \vec{y}).
		\end{equation}
	\end{lemma}
	
	\begin{definition}\label{one_quasi-cyclic_def}
		Let  $g(x)$ and ${{f}_{j}}(x)$ are polynomials in $\mathbb{R}_{q,n}$, $g(x)\mid({{x}^{n}}-1)$, where $j\in[0,\ell-1]$. If $\mathscr{C}$ is a quasi-cyclic code generated by $([g(x){{f}_{0}}(x)]$, $[g(x){{f}_{1}}(x)]$, $\cdots$, $[g(x){{f}_{\ell -1}}(x)])$, then $\mathscr{C}$ is called $1$-generator quasi-cyclic code with index $\ell$.
		The generator matrix $G$ of $\mathscr{C}$  have the following form: 
		\begin{equation}
			G=\left(G_{0}, G_{1}, \cdots, G_{\ell-1}\right) ,
		\end{equation}
		where $G_ {j} $ are $n\times n$  circulant matrices generated by $\left [g (x) f_{j} (x) \right] $.
	\end{definition}

	As a special class of quasi-cyclic codes, $1$-generator quasi-cyclic codes can be regarded as linear codes generated by juxtaposing multiple cyclic codes. The following theorem determines a lower bound on the symplectic distances of $1$-generator quasi-cyclic codes with even index.
	
	\begin{theorem}\label{one_quasi-cyclic}
		Suppose $\mathscr{C}$ is a $1$-generator quasi-cyclic code in Definition \ref{one_quasi-cyclic_def} of index $\ell$. 
		If $gcd(f_j(x)+\alpha f_{j+m}(x),\frac{x^n-1}{g(x)})=1$, and $deg(f_j(x)f_{j+m}(x))\ge 1$, $\alpha\in \mathbb{F}_{q}$, $i\in[0,\ell-1]$, $j\in[0,m-1]$, then the following equation holds.
		\begin{equation}\label{lower_bound}
			d_s(\mathscr{C})\ge m\cdot\left \lceil \frac{q+1}{q} d(g(x))\right \rceil,
		\end{equation}
		where $d_s(\mathscr{C})$ is the symplectic distance of $\mathscr{C}$, $d(g(x))$ is the  minimum Hamming distance of cyclic code $\left \langle g(x) \right \rangle $.
	\end{theorem}
	\begin{proof}
		Given that $a(x)$ is any polynomial in $\mathbb{R}_{q,n}$, then any codeword of $\mathscr{C}$ can be denoted as $\mathbf{c}=([a(x) f_{0}(x) g(x)],[a(x) f_{1}(x) g(x)],$ $\cdots, [a(x) f_{\ell-1}(x) g(x)])$.
		
		Let $\mathbf{c_1}=$ $([a(x) f_{0}(x) g(x)],$ $[a(x) f_{1}(x) g(x)],$ $\cdots,$ $[a(x) $ $f_{m-1}(x) g(x)])$, 
		$\mathbf{c_2}=$ $([a(x) f_{m}(x) g(x)]$, $[a(x) f_{m+1}(x) g(x)]$, $\cdots$, $[a(x) f_{\ell-1}(x) g(x)])$,
		and $\mathbf{c_3}=\mathbf{c_1}+\alpha \mathbf{c_2}= ([a(x)g(x)(f_{0}(x)+\alpha f_{m}(x))],$ $[a(x) g(x)(f_{1}(x)+\alpha f_{m+1}(x))],$
		$\cdots, [a(x) g(x)(f_{m-1}(x)+\alpha f_{\ell-1}(x)) ])$, respectively. 
		By Lemma \ref{symplectic_Hamming}, the symplectic weight of $\mathbf{c}$ is
		
		$$\begin{array}{l}
			\mathrm{w}_{s}(\mathbf{c})
			=\left(\mathrm{w}_{H}(\mathbf{c_1})+\mathrm{w}_{H}(\mathbf{c_2})+\sum\limits_{\alpha \in \mathbb{F}_{q}^{*}} \mathrm{w}_{H}(\mathbf{c_3})\right)/q\\
			= \left(\mathrm{w}_{H}(\mathbf{c_1})+\mathrm{w}_{H}(\mathbf{c_2})\right)/q  \\
			+ \sum\limits_{i=0}^{m} \sum\limits_{\alpha \in \mathbb{F}_{q}^{*}}\mathrm{w}_H([a(x) g(x)(f_{i}(x)+\alpha f_{m+i}(x))]) /q. 
		\end{array}$$
		
		For the reason that $gcd(f_i(x),\frac{x^n-1}{g(x)} )=1$, $i\in[0, \ell-1]$, there are $\mathrm{w}_H(\mathbf{c_1})\ge m\cdot d(g(x))$ and $\mathrm{w}_H(\mathbf{c_2})\ge m\cdot d(g(x))$.
		In addition, it is easy to verify that when $gcd(f_j(x)+\alpha f_{j+m}(x),\frac{x^n-1}{g(x)})=1$, there is $\mathrm{w}_H([a(x) g(x)(f_{i}(x)+\alpha f_{m+i}(x))])\ge d(g(x))$.
		Therefore, the following formula holds.
		$$\begin{array}{rl} 
			\mathrm{w}_{s}(\mathbf{c})\ge&  2m\cdot d(g(x))/q+ (q-1)m\cdot d(g(x))/q\\ 
			\ge&  m\cdot \left \lceil \frac{q+1}{q} d(g(x))\right \rceil.   
		\end{array}$$
	\end{proof}
	\begin{lemma}\label{coprime}
		Let $gcd(n,q)=1$, $g(x)\mid({{x}^{n}}-1)$, and $f(x)$ are polynomials in $\mathbb{R}_{q,n}$.
		If $f(x)\mid g(x)$, then there is $gcd(f(x)+\alpha,\frac{x^n-1}{g(x)})=1$, $\alpha\in \mathbb{F}_{q}$.
	\end{lemma}
	\begin{proof}
		For $gcd(n,q)=1$, $x^n-1$ has no repeated irreducible factors over split field, so if $f(x)\mid g(x)$, then $gcd(f(x),\frac{x^n-1}{g(x)} )=1$.
		In addition, as $f(x)+\alpha$ is not a factor of $\frac{x^n-1}{g(x)}$, $gcd(f(x)+\alpha,\frac{x^n-1}{g(x)})=1$.
	\end{proof}
	
	\begin{corollary}\label{apply_2a}
		When $gcd(n,q)=1$ and $k<\frac{n}{2}$, if there exists a $\mathbb{F}_q$- cyclic code of parameters $[n,k,d]_q$, then there also exist $\mathbb{F}_{q^2}$- additive codes  with parameters $(mn,\frac{k}{2},\ge m\cdot \left \lceil \frac{q+1}{q} d\right \rceil)_{q^2}$.
	\end{corollary}
	\begin{proof}
		It is easy to conclude that Corollary \ref{apply_2a} holds by combining Theorem \ref{one_quasi-cyclic} and Lemma \ref{coprime}.
	\end{proof}
	
	The following lemma will help to determine the optimality of low-dimensional additive codes.
	\begin{lemma}\label{Jopt}
		Let $\mathscr{C}_a$ be an quaternary additive code $(n,k,d)_4$, with $k\ge1$. Then, \begin{equation}\label{AGB}
			3n \geq \sum_{i=0}^{2k-1}\left\lceil\frac{d}{2^{i-1}}\right\rceil.
		\end{equation} 
	\end{lemma}
	\begin{proof}
		The concatenated code of $(n, k, d)_4$ code and $[3,2,2]_{2}$ is $[3n,2k,2d]_2$. By the Griesmer bound, there is Equation (\ref{AGB}), so this lemma holds.
	\end{proof}
	
	\begin{example}
		Let $q=2$, $n=31$, taking generator polynomial $g(x)=x^{26} + x^{24} + x^{22} + x^{21} + x^{20} + x^{18} + x^{17} + x^{13} + x^{12} + x^{11} + x^{10} + x^9 + x^6 + x^5 + x^3 + 1$ from Chen's Database \cite{Chentable}, this will generate a binary cyclic code $\mathscr{C}$ of parameters $[31,5,16]_2$. 
		Through Theorem \ref{one_quasi-cyclic} and Lemma \ref{coprime},
		choosing $f_0(x)=x+1$ and $f_{1}(x)=1$, we can obtain quaternary additive codes with parameters $(31,2.5,\ge 24)_4$. By virtue of Lemma \ref{Jopt},  $(31,2.5, 24)_4$ is optimal additive code, which also has better performance than optimal quaternary linear code with parameters $[31,2, 24]_4$.
		
		It should be noted that similar results were also obtained by Bierbrauer et al. \cite{bierbrauer2021optimal}. However, the approach in this paper is more concise, and our codes have a cyclic (or quasi-cyclic) structure, which makes ours easier to encode and decode.
	\end{example}
	


	\begin{theorem}\label{quasi-cyclic_bound}
		If $\mathscr{C}$ is a $1$-generator quasi-cyclic code  with parameters $[t  n,k,d]_q$ of index $t $. Set polynomials $f_l(x),f_r(x)$ satisfy $gcd(f_l(x)+\alpha f_r(x),\frac{x^n-1}{g(x)})=1$, $\alpha\in \mathbb{F}_{q}$, and $deg(f_l(x)f_{r}(x))\ge 1$;
		then, there also exists an additive code have parameters $(t  n,\frac{k}{2}, \ge \left\lceil \frac{q+1}{q} d \right\rceil)_{q^2}$. 
	\end{theorem}
	\begin{proof}
		Suppose generator of $\mathscr{C}$ is $\mathbf{g(x)}=([g(x){{f}_{0}}(x)]$, $[g(x){{f}_{1}}(x)]$, $\cdots$, $[g(x){{f}_{t -1}}(x)])$. Set $\mathscr{C}^{\prime}$ is a 1-generator quasi-cyclic code with generator 
		$\mathbf{g^{\prime}(x)}=(\mathbf{g(x)}f_l(x)\mid \mathbf{g(x)}f_r(x))$. Let $a(x)$ be any polynomial in $\mathbb{R}_{q,n}$, then any codeword in $\mathscr{C}^{\prime}$ can be denoted as $\mathbf{c}^{\prime}= (a(x)\mathbf{g(x)}f_l(x)\mid a(x)\mathbf{g(x)}f_r(x))$. Let $\mathbf{c_1}^{\prime}=(a(x)\mathbf{g(x)}f_l(x))$, 
		$\mathbf{c_2}^{\prime}=(a(x)\mathbf{g(x)}f_r(x))$,
		and $\mathbf{c_3}^{\prime}=\mathbf{c_1}^{\prime}+ \alpha \mathbf{c_2}^{\prime}$, respectively. 
		With the help of Lemma \ref{symplectic_Hamming}, the symplectic weight of $\mathbf{c}^{\prime}$ is given by the following equation.
		$$\begin{array}{l}
			\mathrm{w}_{s}(\mathbf{c}^{\prime})
			=\left(\mathrm{w}_{H}(\mathbf{c_1}^{\prime})+\mathrm{w}_{H}(\mathbf{c_2}^{\prime})+\sum\limits_{\alpha \in \mathbb{F}_{q}^{*}} \mathrm{w}_{H}(\mathbf{c_3}^{\prime})\right)/q.\\
		\end{array}$$

		Since,  $\mathbf{c_1}^{\prime}, \mathbf{c_2}^{\prime}\in \mathscr{C}$, there are $\mathrm{w}_H(\mathbf{c_1}^{\prime}) \ge d$, and $\mathrm{w}_H(\mathbf{c_2}^{\prime}) \ge d$. 
		In addition, for the reason that $\forall \alpha  \in \mathbb{F}_q$, $gcd(f_l(x)+\alpha f_r(x),\frac{x^n-1}{g(x)})=1$, and $\mathbf{c_3}^{\prime}=a(x)(f_l(x)+\alpha f_r(x))(\mathbf{g(x)})$; hence, $\sum\limits_{\alpha \in \mathbb{F}_{q}^{*}} \mathrm{w}_{H}(\mathbf{c_3^{\prime}}) \ge (q-1)d$.
		Therefore, the following formula holds.
		$$\begin{array}{rl} 
			\mathrm{w}_{s}(\mathbf{c}^{\prime})\ge&  2  d/q+ (q-1)  d/q\\ 
			\ge&  \left \lceil \frac{q+1}{q} d\right \rceil.   
		\end{array}$$
		
		Then it is clear that $\mathscr{C}^{\prime}$ is an symplectic code have parameters $\left[2t  n,k,\ge \left \lceil \frac{q+1}{q} d\right \rceil\right]_q^s$, which corresponds to an additive $\left(t  n,\frac{k}{2}, \ge\left\lceil \frac{q+1}{q} d \right\rceil\right)_{q^2}$ code.

	\end{proof}
	\begin{example}
		Let $q=2$, $n=127$, taking polynomial $h(x)=x^7 + x^6 + x^5 + x^3 + x^2 + x + 1$, then $g(x)=\frac{x^n-1}{h(x)}$ will generate an optimal binary cyclic code $\mathscr{C}$ of parameters $[127,7,64]_2$.  
		Through Lemma \ref{coprime},
		choosing $f_0(x)=x+1$ and $f_{1}(x)=1$; then $([g(x)f_0(x)],[g(x)f_1(x)])$ can generate a quasi-cyclic code $\mathscr{C}_l$ with parameters $[254,7,128]_2$. By Theorem \ref{one_quasi-cyclic} and Lemma \ref{Jopt}, $\Phi(\mathscr{C}_l)$ is an optimal quaternary $(127,3.5, 96)_4$\footnote{This code is also obtained by Guo et al. in \cite{guo2017construction}, but our construction is simpler and has a cyclic structure.}, which has better performance than optimal quaternary linear codes $[127,3, 96]_4$ in \cite{Grassltable}. In addition, with Theorem \ref{quasi-cyclic_bound}, taking $f_l(x)=f_0(x)$, $f_r(x)=1$; then, $([g(x)f_0(x)f_l(x)]$, $[g(x)f_1(x)f_l(x)]$, $[g(x)f_0(x)f_r(x)]$, $[g(x)f_1(x)f_r(x)])$ will generate a $[508,7,\ge 192]_2^s$ symplectic code. So, by Lemma \ref{Jopt}, there exist optimal $(254,3.5,192)_4$ additive code, which also outperform optimal quaternary linear codes $[254,3, 192]_4$. 
		
	\end{example}
	\begin{remark}
		Lemma \ref{coprime} gives only one way to select $f_i(x)$ that satisfies Theorem \ref{one_quasi-cyclic} and \ref{quasi-cyclic_bound}. This ensures that the distance of the resulting additive code is greater than or equal to the lower bound but is not necessarily the best. For additive codes of dimension greater than 3.5, a computational search would be an efficient way to construct additive codes with suitable parameters.
		An efficient search approach is to choose the generator polynomial of best cyclic code as $g(x)$. It is more probable to find codes with considerable distance even if the lower bound of the distance derived from Theorem \ref{one_quasi-cyclic} and \ref{quasi-cyclic_bound} is not large. Similar to the ASR algorithm, proposed by Aydin, Siap and Ray-Chaudhuri in \cite{aydin2001structure}, this can help us search for good additive codes.
	\end{remark}
	\section{Good quaternary additive codes outperform best-known linear codes}\label{IV}
	This section focuses on constructing good additive codes. 
	We construct many additive codes superior to linear counterparts, specifically, some of which perform better than best-known quaternary linear codes in \cite{Grassltable}. Moreover, employing a combinatorial approach, we partially determine the parameters of optimal quaternary additive 3.5-dimensional codes of lengths from $28$ to $254$.
	
	Consider the finite field of order $2$ by $\mathbb{F}_2$ and the finite field of order $4$ by $\mathbb{F}_4 = \{0, 1, w, w^2\}$, where $w^2 +w + 1 = 0$. In addition, $\mathbf{1_n}$, $\mathbf{w_n}$, $\mathbf{w^2_n}$ denote all $1$, $w$ and $w^2$ vectors of length $n$, respectively. 
	
	Before starting construction, we introduce additive codes' basic derivation and augmentation.
	\begin{lemma}
		If $\mathscr{C}_a$ is an additive code of parameters $(n,k,d)_{q^2}$, then the following additive codes also exist:\par
		(1) For $i\ge 1$, $(n+i,k,\ge d)_{q^2}$ (Additive Extend); \par
		(2) For $i\le d$, $(n-i,k,\ge d-i)_{q^2}$ (Additive Puncture); \par
		(3) For $i \le k$, $(n-i,k-i,\ge d)_{q^2}$ (Additive Shorten); 
	\end{lemma}
	\begin{proof}
		The difference between generator matrix of additive and linear codes is that the number of rows of additive codes is $2k$. For the reason that puncture and extension can be considered as the extensions and deletions of the codewords. Therefore, the extent and puncture of additive and linear codes have similar properties, so (1) and (2) naturally hold.
		Since the proof of Theorem 1.5.7 (i) in \cite{huffman2010fundamentals} is also valid for additive codes, the puncture of $\mathscr{C}_a$ is equivalent to the shorten of $\mathscr{C}_a^{\perp}$.
		By puncturing $\mathscr{C}_a^{\perp}$, one can get an additive code $(n-1,k)_{q^2}$, which satisfies that at least any $d-1$ columns are linearly independent. Hence, there exist $(n-i,k-i,\ge d)_{q^2}$ additive codes. Therefore, (3) also holds.
	\end{proof}
	%

	\begin{lemma}(Additive Augmentation) \label{Augment}
		If $\mathscr{C}_a$ is an additive code of parameters $(n,k)_{q^2}$, and no codeword of weight $n$ in $\mathscr{C}_a$, then the following additive codes exist.\par
		(1) $(n,k+0.5)_{q^2}$.\par
		(2) $(n,k+1)_{q^2}$.
	\end{lemma}
	\begin{proof}
		For $\mathbb{F}_{q^2}^n$, there are three special vectors of weight $n$, $\mathbf{1_n}$, $\mathbf{w_n}$, $\mathbf{w^2_n}$, which can be expressed in two bases $\mathbf{1_n}$ and $\mathbf{w_n}$. 
		Since, there is no codeword in $\mathscr{C}_a$ of weight $n$, then $\mathbf{1_n}$ and $\mathbf{w_n}$ both cannot be spanned by bases of $\mathscr{C}_a$. 
		Therefore, adding $\mathbf{1_n}$, one can get an $(n,k+0.5)_{q^2}$ additive code; adding $\mathbf{1_n}$ and $\mathbf{w_n}$ to $\mathscr{C}_a$, one can get an $(n,k+1)_{q^2}$ additive code. 
	\end{proof}
	
	\begin{example}
		Let $q=2$, $n=63$. Taking $g(x)=x^{53} + x^{52} + x^{51} + x^{50} + x^{48} + x^{47} + x^{45} + x^{43} + x^{42} + x^{40} + x^{39} +x^{38} + x^{31} + x^{28} + x^{25} + x^{24} + x^{21} + x^{20} + x^{19} + x^{17} + x^{14} + x^{13} + x^9 + x^8 + x^5 + x + 1$, $\langle g(x) \rangle $ can generator a $[63,10,27]_2$ cyclic code, selecting
		$f_0(x)=x^{61} + x^{59} + x^{58} + x^{54} + x^{52} + x^{50} + x^{45} + x^{44} + x^{43} + x^{41} + x^{39} + x^{33} + x^{32} + x^{31} + x^{26} + x^{24} + x^{23} + x^{22} + x^{20} + x^{18} + x^{13} + x^{12} + x^{11} + x^{10} + x^9 + x^6 + x^4 + x^2 + x$,
		and $f_1(x)=1$. Since $gcd(f_0(x)+ f_{1}(x),\frac{x^n-1}{g(x)})=1$, and $deg(f_0(x)f_{1}(x))\ge 1$, $([g(x)f_0(x)],[g(x)])$ will generate a symplectic $[126,10,\ge 41]_2^s$ code $\mathscr{C}_s$.  
		Using Magma \cite{bosma1997magma}, the real distance of $\mathscr{C}_a=\Phi(\mathscr{C}_s)$ can be calculated as $45$. Therefore, we can get an additive code of parameters $(63,5,45)_4$, which have a larger minimum distance compared with best-known linear code $[63,5,44]_4$ in \cite{Grassltable}.  
		
		In addition, augment $(63,5,45)_4$, we can get $(63,5.5,45)_4$; extend $(63,5,45)_4$, we can get $(64,5,46)_4$. Extend $(63,5.5,45)_4$, we can get $(64,5.5,46)_4$. Best-known linear codes in \cite{Grassltable} are $[63,5,44]_4$, $[64,5,45]_4$, so $(63,5.5,45)_4$, $(64,5,46)_4$ and $(64,5.5,46)_4$ all outperform best-known linear counterparts.
	\end{example}

	\begin{lemma}(Additive Construction X)\label{construction_X}
		If there are two additive codes $\mathscr{C}_{a2} \subset \mathscr{C}_{a1}$, with parameters $(n,k_2,d_2)_{q^2} \subset (n,k_1,d_1)_{q^2}$, where $d_2 > d_1$. Let $\mathscr{C}_{a3}$ be an additive code with parameters $(l,k_1-k_2,\delta )_{q^2}$, then there exists $(n+l,k_1,\min\{\delta+d_1,d_2\} )_{q^2}$ additive code.
	\end{lemma}
	\begin{proof}
		First, suppose the generator matrices of $\mathscr{C}_{a1}$, $\mathscr{C}_{a2}$ and $\mathscr{C}_{a3}$ as $G_{a1}$, $G_{a2}$ and $G_{a3}$, respectively. 
		Since, $\mathscr{C}_{a2} \subset \mathscr{C}_{a1}$, there is $G_{a1}=\begin{pmatrix}
			G_{a2} \\
			G_{ax}
		\end{pmatrix}$.
		Then, constructing $Gx=\begin{pmatrix}
			G_{ax}	& G_{a3} \\
			G_{a2}	&0
		\end{pmatrix}$, which can generate an additive code with parameters $(n+l,k_1,\min\{\delta+d_1,d_2\} )_{q^2}$.
	\end{proof}
	\begin{example}
		Let $q=2$, $n=35$. Taking $g(x)=x^{28}+ x^{25}+ x^{24}+ x^{21}+ x^{20}+ x^{19}+ x^{18}+ x^{17}+ x^{15}+ x^{13}+ x^{12}+ x^{11}+ x^{9}+ x^{8}+ x^{6}+ x^{2}+ x+ 1$, $f_0(x)=x^{34}+ x^{33}+ x^{30}+ x^{28}+ x^{27}+ x^{25}+ x^{23}+ x^{22}+ x^{18}+ x^{17}+ x^{16}+ x^{11}+ x^{8}+ x^{7}+ x^{6}+ x^{5}+ x^{4}$, $f_1(x)=x^{34}+ x^{27}+ x^{23}+ x^{21}+ x^{20}+ x^{19}+ x^{18}+ x^{16}+ x^{15}+ x^{13}+ x^{11}+ x^{10}+ x^{5}+ x^{4}+ x$.
		Then, $([g(x)f_0(x)],[g(x)f_1(x)])$ will generate a symplectic $[70,7,26]_2^s$ code, which corresponds to an additive code of parameters $(35,3.5,26)_4$. Further, we can get a subcode with parameters $(35,1.5,30)_4$ by taking out codewords of weight $30$ from $(35,3.5,26)_4$. Generator matrix of additive $(35,1.5,30)_4$ code is $G_{sub}$. 
		
		\begin{figure*}[htbp] 
			\centering
			\begin{equation*}
				\setlength{\arraycolsep}{0.1pt}
				G_{sub}= \left( {\begin{array}{*{35}{c}}
						1 & w & 1 & w^2 & w^2 & w   & 0   & 1 & w & 1 & w^2 & w^2 & w   & 0   & 1 & w & 1 & w^2 & w^2 & w   & 0   & 1 & w & 1 & w^2 & w^2 & w   & 0   & 1 & w & 1 & w^2 & w^2 & w   & 0   \\
						w & 0 & 1 & w   & 1   & w^2 & w^2 & w & 0 & 1 & w   & 1   & w^2 & w^2 & w & 0 & 1 & w   & 1   & w^2 & w^2 & w & 0 & 1 & w   & 1   & w^2 & w^2 & w & 0 & 1 & w   & 1   & w^2 & w^2 \\
						0 & 1 & w & 1   & w^2 & w^2 & w   & 0 & 1 & w & 1   & w^2 & w^2 & w   & 0 & 1 & w & 1   & w^2 & w^2 & w   & 0 & 1 & w & 1   & w^2 & w^2 & w   & 0 & 1 & w & 1   & w^2 & w^2 & w
				\end{array}} \right).
			\end{equation*}
		\end{figure*}
		
		According to Lemma \ref{construction_X}, select $(5,2,4)_4$ as auxiliary code which leads to an optimal additive code with parameters $(40,3.5,30)_4$. 
		This code performs better than optimal linear code $[40,3,30]_4$ in \cite{Grassltable}.  
	\end{example}

	It should be noted that Theorem \ref{one_quasi-cyclic} also has implications for constructing good additive codes via the symplectic multi-generator quasi-cyclic codes.
	
	\begin{example}
		Let $q=2$, $n=47$. Selecting $g_1(x)=x^{23} + x^{19} + x^{18} + x^{14} + x^{13} + x^{12} + x^{10} + x^9 + x^7 + x^6 + x^5 + x^3 + x^2 + x + 1$ and $g_2(x)=x + 1$, which can generate two optimal binary cyclic codes with parameters $[47,24,11]_2$ and $[47,46,2]_2$, respectively; 
		By computer search, we find $f(x)=x^{45} + x^{44} + x^{43} + x^{42} + x^{41} + x^{39} + x^{38} + x^{37} + x^{33} + x^{32} + x^{15} + x^{14} + x^{10} + x^9 + x^8 + x^6 + x^5 + x^4 +  x^3 + x^2$  such that the parameters of the 2-generator quasi-cyclic code with $([g_1(x)f(x),g_1(x)])$ and $([g_2(x),g_2(x)f(x)])$ as generators are $(47,35,7)_4$. This result surpasses best-known quaternary linear code $[47,35,6]_4$ in Grassl's code table \cite{Grassltable}.
	\end{example}
	\begin{example}
		Let $q=2$, $n=55$. Selecting $g_1(x)=x^{35} + x^{34} + x^{32} + x^{31} + x^{30} + x^{28} + x^{27} + x^{26} + x^{24} + x^{23} + x^{20} +x^{18} + x^{14} + x^{13} + x^{11} + x^{10} + x^{8} + x^{7} + x^{6} + x^{5} + x^{4} + x^{3} + x^{2} + 1$ and $g_2(x)=\frac{x^n-1}{x+1}$, which can generate two optimal binary cyclic codes with parameters $[55,20,16]_2$ and $[55,1,55]_2$, respectively. 
		By computer search, we find $f_0(x)=x^{50} + x^{49} + x^{47} + x^{43} + x^{42} + x^{41} + x^{39} + x^{38} + x^{37} + x^{36} + x^{35} + x^{34} + x^{33} + x^{31} + x^{29} + x^{25} + x^{24} + x^{21} + x^{20} + x^{16} + x^{15} + x^{13} + x^{12} + x^{10} + x^{9} + x^{8} + x^{5} + x^{2} + x$;
		$f_1(x)=x^{53} + x^{52} + x^{51} + x^{50} + x^{49} + x^{47} + x^{46} + x^{45} + x^{44} + x^{40} + x^{39} +x^{37} + x^{36} + x^{33} + x^{30} + x^{29} + x^{27} + x^{26} + x^{25} + x^{24} + x^{22} + x^{21} + x^{20} + x^{18} + x^{17} + x^{16} + x^{15} + x^{14} + x^{13} + x^{9} + x^{7} + x^{5} + x^4$;  such that the parameters of the 2-generator quasi-cyclic code with $([g_1(x)f_0(x),g_1(x)f_1(x)])$ and $([g_2(x)f_1(x),g_2(x)f_0(x)])$ as generators are $(55,10.5,29)_4$. 
		
		Extending $(55,10.5,29)_4$, we can get $(56,10.5,30)_4$, which fill the gap of $[56,10,32]_4$ and $[56,11,29]_4$. In addition, augmenting $(56,10.5,30)_4$, one can get an additive code have parameters $(56,11,30)_4$.
		This result is surpassing best-known quaternary linear code $[56,11,29]_4$ in Grassl's code table \cite{Grassltable}.
	\end{example}

	%
	%
	
	\begin{remark}
		According to Lemma \ref{Augment}, by augmenting the additive codes $(31,2.5,24)_4$ and $(127,3.5,96)_4$ in Example 1 and Example 2 twice, we can obtain $(31,3.5,23)_4$ and $(127,4.5,95)_4$.
		Further, extending $(31,3.5,23)_4$ and $(127,4.5,95)_4$, we can also get $(32,3.5,24)_4$ and $(128,4.5,96)_4$. 
		In particular, $(31,3.5,23)_4$, $(32,3.5,24)_4$, $(127,4.5,95)_4$ and $(128,4.5,96)_4$ all perform better than the best linear codes $[31,3,23]_4$, $[32,3,24]_4$, $[127,4,95]_4$ and $[128,4,96]_4$ in \cite{Grassltable}.    
		
		In addition, due to the extension of additive codes obeying even-like principle, so there exists $(32,1,32)_4\subset (32,3.5,24)_4$, selecting $(11,2.5,8)_4$ as auxiliary code $\mathscr{C}_{aux}$. Generator matrix $\mathscr{C}_{aux}$ is given here.
		\setlength{\arraycolsep}{0.5pt}
		$$G_{aux}= \left( {\begin{array}{*{11}{c}}
				1&0&w&1&1&w&0&w&1&w&0\\
				w&0&w&1&w&0&w&1&w&0&w\\
				0&1&0&1&w^2&0&w^2&1&w^2&w^2&1\\
				0&w&w&0&1&w^2&w^2&1&0&w&w\\
				0&0&w^2&w&1&1&w&0&w&w&w^2
		\end{array}} \right).$$

		Then, with Lemma \ref{construction_X}, we can get an additive $(43,3.5,32)_4$ code. Similarly, there also has $(128,1,128)_4$$\subset$ $(128,4.5,96)_4$, so separately selecting additive codes $(32,3.5,24)_4$, $(35,3.5,26)_4$, $(40,3.5,30)_4$ and $(43,3.5,32)_4$ as auxiliary codes $\mathscr{C}_{aux}^{\prime}$; then, with Lemma \ref{construction_X}, we can get optimal additive codes with parameters $(160,4.5,120)_4$, $(163,4.5,122)_4$, $(168,4.5,126)_4$, $(171,4.5,128)_4$.  By virtue of Lemma \ref{Jopt}, $(43,3.5,32)_4$, $(160,4.5,120)_4$, $(163,4.5,122)_4$, $(168,4.5,126)_4$, $(171,4.5,128)_4$ are all optimal additive codes and have better performance than optimal linear codes in \cite{Grassltable}.
	\end{remark}
	
	\begin{remark} \label{reviewer}
		We would like to thank the reviewer for the professional reminder that the extremal code $[85, 4, 64]_4$ in $P G(3,4)$ reveals the existence of $(85,3.5,64)_4$.
	\end{remark}
	
	A combination method can obtain more optimal $3.5$-dimensional additive codes ranging from $28$ to $254$,  as shown in Table 1.
	For clarity of presentation, we make $\mathscr{C}_t$ denote quaternary additive $(n,3.5,n-t)_4$ code. $(\mathscr{C}_{t_1}\mid\mathscr{C}_{t_2})$ denotes the juxtaposition code of $\mathscr{C}_{t_1}$ and $\mathscr{C}_{t_2}$. Moreover, when $\mathscr{C}_{t_1}$ and $\mathscr{C}_{t_2}$ are combined, we default them both to the maximum length. The fourth column of Table \ref{3.5_additive} shows the range of the optimal additive codes derived with Lemma \ref{Jopt}. In addition, most perform better than the optimal 3-dimensional linear codes in \cite{Grassltable}, but we will not list them here due to space limitations.

	\begin{table}[htbp]
		\caption{Optimal quaternary additive 3.5-dimensional codes of lengths form $28$ to $254$}\label{3.5_additive}
		\centering
		\begin{tabular}{ccccc}
			\toprule
			No. & $\mathscr{C}_{t}$  &    Parameters    &       Range        &               Constructions                \\ \midrule
			1  & $\mathscr{C}_{8}$  & $(n,3.5,n-8)_4$  &  $28\le n \le 32$  &                 Remark 2                  \\
			2  & $\mathscr{C}_{9}$  & $(n,3.5,n-9)_4$  &  $33\le n \le 35$  &                 Example 4                 \\
			3  & $\mathscr{C}_{10}$ & $(n,3.5,n-10)_4$ &  $36\le n \le 40$  &                 Example 4                 \\
			4  & $\mathscr{C}_{11}$ & $(n,3.5,n-11)_4$ &  $41\le n \le 43$  &                 Remark 2                  \\
			5  & $\mathscr{C}_{16}$ & $(n,3.5,n-16)_4$ &  $60\le n \le 64$  &  $(\mathscr{C}_{8}\mid\mathscr{C}_{8})$   \\
			6  & $\mathscr{C}_{17}$ & $(n,3.5,n-17)_4$ &  $65\le n \le 67$  &  $(\mathscr{C}_{8}\mid\mathscr{C}_{9})$   \\
			7  & $\mathscr{C}_{18}$ & $(n,3.5,n-18)_4$ &  $68\le n \le 72$  &  $(\mathscr{C}_{8}\mid\mathscr{C}_{10})$   \\
			8  & $\mathscr{C}_{19}$ & $(n,3.5,n-19)_4$ &  $73\le n \le 75$  &  $(\mathscr{C}_{9}\mid\mathscr{C}_{10})$  \\
			9  & $\mathscr{C}_{20}$ & $(n,3.5,n-20)_4$ &  $76\le n \le 80$  & $(\mathscr{C}_{10}\mid\mathscr{C}_{10})$  \\
			10  & $\mathscr{C}_{21}$ & $(n,3.5,n-21)_4$ &  $81\le n \le 85$  & Remark \ref{reviewer}  \\
			11  & $\mathscr{C}_{22}$ & $(n,3.5,n-22)_4$ &  $86\le n \le 86$  & $(\mathscr{C}_{11}\mid\mathscr{C}_{11})$  \\
			12  & $\mathscr{C}_{24}$ & $(n,3.5,n-24)_4$ &  $92\le n \le 96$  &  $(\mathscr{C}_{8}\mid\mathscr{C}_{16})$  \\
			13  & $\mathscr{C}_{25}$ & $(n,3.5,n-25)_4$ &  $97\le n \le 99$  &  $(\mathscr{C}_{9}\mid\mathscr{C}_{16})$  \\
			14  & $\mathscr{C}_{26}$ & $(n,3.5,n-26)_4$ & $102\le n \le 104$ & $(\mathscr{C}_{10}\mid\mathscr{C}_{16})$  \\
			15  & $\mathscr{C}_{27}$ & $(n,3.5,n-27)_4$ & $107\le n \le 107$ &  $(\mathscr{C}_{8}\mid\mathscr{C}_{19})$  \\
			16  & $\mathscr{C}_{28}$ & $(n,3.5,n-28)_4$ & $108\le n \le 112$ &  $(\mathscr{C}_{8}\mid\mathscr{C}_{20})$  \\
			17  & $\mathscr{C}_{29}$ & $(n,3.5,n-29)_4$ & $113\le n \le 115$ &  $(\mathscr{C}_{9}\mid\mathscr{C}_{20})$  \\
			18  & $\mathscr{C}_{30}$ & $(n,3.5,n-30)_4$ & $118\le n \le 120$ & $(\mathscr{C}_{10}\mid\mathscr{C}_{20})$  \\
			19  & $\mathscr{C}_{31}$ & $(n,3.5,n-31)_4$ & $123\le n \le 127$ &                 Example 2                 \\
			20  & $\mathscr{C}_{32}$ & $(n,3.5,n-32)_4$ & $128\le n \le 128$ &         Extend $\mathscr{C}_{31}$         \\
			21  & $\mathscr{C}_{34}$ & $(n,3.5,n-34)_4$ & $134\le n \le 136$ & $(\mathscr{C}_{10}\mid\mathscr{C}_{24})$  \\
			22  & $\mathscr{C}_{39}$ & $(n,3.5,n-39)_4$ & $155\le n \le 159$ & $(\mathscr{C}_{8}\mid\mathscr{C}_{31})$  \\
			23  & $\mathscr{C}_{40}$ & $(n,3.5,n-40)_4$ & $160\le n \le 162$ & $(\mathscr{C}_{9}\mid\mathscr{C}_{31})$  \\
			24  & $\mathscr{C}_{41}$ & $(n,3.5,n-41)_4$ & $163\le n \le 167$ & $(\mathscr{C}_{10}\mid\mathscr{C}_{31})$ \\
			25  & $\mathscr{C}_{42}$ & $(n,3.5,n-42)_4$ & $168\le n \le 170$ & $(\mathscr{C}_{11}\mid\mathscr{C}_{31})$ \\
			26  & $\mathscr{C}_{47}$ & $(n,3.5,n-47)_4$ & $187\le n \le 191$ & $(\mathscr{C}_{16}\mid\mathscr{C}_{31})$ \\
			27  & $\mathscr{C}_{48}$ & $(n,3.5,n-48)_4$ & $192\le n \le 194$ & $(\mathscr{C}_{17}\mid\mathscr{C}_{31})$ \\
			28  & $\mathscr{C}_{49}$ & $(n,3.5,n-49)_4$ & $195\le n \le 199$ & $(\mathscr{C}_{18}\mid\mathscr{C}_{31})$ \\
			29  & $\mathscr{C}_{50}$ & $(n,3.5,n-50)_4$ & $200\le n \le 202$ & $(\mathscr{C}_{19}\mid\mathscr{C}_{31})$ \\
			30  & $\mathscr{C}_{51}$ & $(n,3.5,n-51)_4$ & $203\le n \le 207$ & $(\mathscr{C}_{20}\mid\mathscr{C}_{31})$ \\
			31  & $\mathscr{C}_{52}$ & $(n,3.5,n-52)_4$ & $208\le n \le 212$ & $(\mathscr{C}_{21}\mid\mathscr{C}_{31})$ \\
			32  & $\mathscr{C}_{53}$ & $(n,3.5,n-53)_4$ & $213\le n \le 213$ & $(\mathscr{C}_{22}\mid\mathscr{C}_{31})$ \\
			33  & $\mathscr{C}_{55}$ & $(n,3.5,n-55)_4$ & $219\le n \le 223$ & $(\mathscr{C}_{24}\mid\mathscr{C}_{31})$ \\
			34  & $\mathscr{C}_{56}$ & $(n,3.5,n-56)_4$ & $224\le n \le 226$ & $(\mathscr{C}_{25}\mid\mathscr{C}_{31})$ \\
			35  & $\mathscr{C}_{57}$ & $(n,3.5,n-57)_4$ & $229\le n \le 231$ & $(\mathscr{C}_{26}\mid\mathscr{C}_{31})$ \\
			36  & $\mathscr{C}_{58}$ & $(n,3.5,n-58)_4$ & $234\le n \le 234$ & $(\mathscr{C}_{27}\mid\mathscr{C}_{31})$ \\
			37  & $\mathscr{C}_{59}$ & $(n,3.5,n-59)_4$ & $235\le n \le 239$ & $(\mathscr{C}_{28}\mid\mathscr{C}_{31})$ \\
			38  & $\mathscr{C}_{60}$ & $(n,3.5,n-60)_4$ & $240\le n \le 242$ & $(\mathscr{C}_{29}\mid\mathscr{C}_{31})$ \\
			39  & $\mathscr{C}_{61}$ & $(n,3.5,n-61)_4$ & $245\le n \le 247$ & $(\mathscr{C}_{30}\mid\mathscr{C}_{31})$ \\
			40  & $\mathscr{C}_{62}$ & $(n,3.5,n-62)_4$ & $250\le n \le 254$ & $(\mathscr{C}_{31}\mid\mathscr{C}_{31})$ \\ \bottomrule
		\end{tabular}%

	\end{table}
	
	Furthermore, we also construct a number of good additive codes, all of which are better than their linear counterparts in \cite{Grassltable}. We give their specific constructions in Table \ref{good_QC} in Appendix and compare their parameters with best-known linear codes in \cite{Grassltable} in Tables \ref{good_additive} and \ref{higher_rate_additive} to illustrate the effectiveness of the methods in this paper. In particular, Table \ref{good_additive} is a comparison done for the same length $n$ and dimension $k$. Our additive codes have a greater distance. Table  \ref{higher_rate_additive} lists additive codes with higher information rates that are compared with best-known linear codes, i.e., they have twice as many codewords as the corresponding linear codes for the same code length $n$ and distance $d$.

	\begin{table}[htbp]
		\caption{Quaternary Additive codes outperform linear counterparts}\label{good_additive}
		\centering
		\begin{tabular}{ccc}
			\toprule
			No. & Our Additive Codes & Linear Counterparts in \cite{Grassltable} \\ \midrule
			1  &   $(47,35,7)_4$    &               $[47,35,6]_4$               \\
			2  &    $(56,11,30)$    &              $[56,11,29]_4$               \\
			3  &   $(64,5,46)_4$    &               $[64,5,45]_4$               \\
			4  &   $(92,7,64)_4$    &               $[92,7,62]_4$               \\
			5  &  $(128,11,80)_4$   &              $[128,11,79]_4$              \\
			6  &  $(196,7,140)_4$   &              $[196,7,139]_4$              \\ \bottomrule
		\end{tabular}%
	\end{table}
	
	\begin{table}[htbp]
		\centering
		
		\caption{Quaternary Additive codes have higher information rates than linear counterparts}\label{higher_rate_additive}
		\begin{threeparttable}
			
			\begin{tabular}{cccc}
				\toprule
				No. &        $(n,k,d)_4$         & $[n,\lfloor k \rfloor,d]_4$ & $[n,\lceil k \rceil,d]_4$ \\ \midrule
				1  &      $(21,10.5,8)_4$       &        $[21,10,8]_4$        &       $[21,11,7]_4$       \\
				2  &      $(22,12.5,7)_4$       &        $[22,12,7]_4$        &       $[22,13,6]_4$       \\
				3  &      $(22,15.5,5)_4$       &        $[22,15,5]_4$        &       $[22,16,4]_4$       \\
				4  &      $(30,6.5,18)_4$       &        $[30,6,18]_4$        &       $[30,7,17]_4$       \\
				5  &      $(35,9.5,18)_4$       &        $[35,9,18]_4$        &      $[35,10,17]_4$       \\
				6  &      $(42,24.5,10)_4$      &       $[42,24,10]_4$        &       $[42,25,9]_4$       \\
				7  &      $(45,14.5,19)_4$      &       $[45,14,19]_4$        &      $[45,15,18]_4$       \\
				8  &      $(44,13.5,19)_4$      &       $[44,13,19]_4$        &      $[44,14,18]_4$       \\
				9  &      $(45,15.5,18)_4$      &       $[45,15,18]_4$        &      $[45,16,17]_4$       \\
				10  &      $(45,16.5,17)_4$      &       $[45,16,17]_4$        &      $[45,17,16]_4$       \\
				11  &      $(47,11.5,24)_4$      &       $[47,11,24]_4$        &      $[47,12,23]_4$       \\
				12  &      $(63,5.5,45)_4$       &        $[63,5,44]_4$        &       $[63,6,44]_4$       \\
				13  &      $(64,5.5,46)_4$       &        $[64,5,45]_4$        &       $[64,6,44]_4$       \\
				14  &      $(71,17.5,32)_4$      &       $[71,17,32]_4$        &      $[71,18,31]_4$       \\
				15  &      $(73,13.5,38)_4$      &       $[73,13,38]_4$        &      $[73,14,37]_4$       \\
				16  &   $(124,4.5,92)_4^\star$   &       $[124,4,92]_4$        &      $[124,5,90]_4$       \\
				17  &   $(127,4.5,95)_4^\star$   &       $[127,4,95]_4$        &      $[127,5,92]_4$       \\
				18  &   $(128,4.5,96)_4^\star$   &       $[128,4,96]_4$        &      $[128,5,93]_4$       \\
				19  &  $(155,4.5,116)_4^\star$   &       $[155,4,116]_4$       &      $[155,5,114]_4$      \\
				20  &  $(160,4.5,120)_4^\star$   &       $[160,4,120]_4$       &      $[160,5,118]_4$      \\
				21  &  $(163,4.5,122)_4^\star$   &       $[163,4,122]_4$       &      $[163,5,120]_4$      \\
				22  &  $(166,4.5,124)_4^\star$   &       $[166,4,124]_4$       &      $[166,5,122]_4$      \\
				23  &  $(168,4.5,126)_4^\star$   &       $[168,4,126]_4$       &      $[168,5,124]_4$      \\
				24  &  $(171,4.5,128)_4^\star$   &      $[171,4.5,128]_4$      &      $[171,5,127]_4$      \\ \bottomrule
			\end{tabular}%
			\begin{tablenotes}    
				\footnotesize               
				\item[$\star$] According to Lemma \ref{Jopt}, corresponding additive codes are optimal.        
			\end{tablenotes}            
		\end{threeparttable}     
	\end{table}

	\section{Some quaternary ACD codes have larger distance than LCD codes}\label{V}
	This section focuses on constructing good ACD codes with greater distance than best-known quaternary LCD codes in \cite{lu2020optimal,harada2021construction,Ishizuka2022,ishizuka2022construction}.

	In \cite{Guan2023OnEH}, the authors identify sufficient and necessary conditions for the quasi-cyclic codes to be symplectic LCD, as shown in Lemma \ref{symplectic_LCD}.
	
	\begin{lemma}(\cite{Guan2023OnEH})\label{symplectic_LCD}
		Let $\mathscr{C}$ be a $1$-generator quasi-cyclic code in Definition \ref{one_quasi-cyclic_def} of index $\ell$.
		Taking $\Lambda =\sum\limits_{j=0}^{m-1}({{f}_{j}}(x){{\bar{f}}_{m+j}}(x)-{{f}_{m+j}}(x){{\bar{f}}_{j}}(x))$, then $\mathscr{C}$ is symplectic LCD  code if and only if the following equations hold.
		
		\begin{equation}
			\begin{array}{c} 
				g(x)=\tilde{g}(x), \\ 
				gcd(\Lambda,\frac{x^n-1}{g(x)})=1.
			\end{array}
		\end{equation}	
	\end{lemma}

	\begin{corollary}\label{one_quasi-cyclic_sLCD}
		Let $\mathscr{C}$ be a $1$-generator quasi-cyclic code in Definition \ref{one_quasi-cyclic_def} of index $\ell$.
		If generator of $\mathscr{C}$ satisfying Theorem \ref{one_quasi-cyclic} and Lemma \ref{symplectic_LCD}, then there exists a symplectic LCD code with the following parameters:
		\begin{equation}\nonumber
			\left[\ell n,n-deg(g(x)),\ge m\cdot\left \lceil \frac{q+1}{q} d(g(x))\right \rceil\right]_q^s.
		\end{equation}
	\end{corollary}
	
	\begin{remark}
		Since binary symplectic LCD codes are equivalent to quaternary ACD codes, so corollary \ref{one_quasi-cyclic_sLCD} also reveals the existence of ACD codes with parameters:
		\begin{equation}\nonumber
			\left(m n,(n-Degree(g(x)))/2,\ge m\cdot\left \lceil \frac{q+1}{q} d(g(x))\right \rceil\right)_{q^2}.
		\end{equation} 
	\end{remark}
	
	\begin{example}
		Let $q=2$, $n=13$. Taking $g(x)=x + 1$, which can generate an optimal binary LCD code with parameters $[13, 12, 2]_2$. Selecting 
		$f_1(x)=x^{12} + x^9 + x^8 + x^7 + x^6 + x^5 + x^4 + x^3$, $f_2(x)=x^{12} + x^9 + x^8 + x^6 + x$, and $f_3(x)=x^{10} + x^9 + x^8 + x^7 + x^6 + x^3 + x^2 + x + 1$; 
		Then, $([g(x)],[g(x)f_1(x)],[g(x)f_2(x)],[g(x)f_3(x)])$ will generate a symplectic $[52,12,\ge 6]_2^s$ LCD code. Using Magma \cite{bosma1997magma}, one can compute the real symplectic distance of this code as $15$. Therefore, we can get an ACD code of parameters $(26,6,15)_4$, which have a larger minimum distance compared with best-known LCD code $[26,6,14]_4$ in \cite{ishizuka2022construction}.
	\end{example}
	
	The following lemma will be helpful in constructing new ACD codes by combining trace Hermitian self-orthogonal codes and ACD codes.
	\begin{lemma}\label{construction_AB}
		If there exist $(n_1,k,d_1)_{q^2}$ ACD code $\mathscr{C}_{a1}$ and $(n_2,k,d_2)_{q^2}$ additive trace Hermitian self-orthogonal code $\mathscr{C}_{a2}$, respectively. Then, there also exists $(n_1+n_2,k,\ge(d_1+d_2))_{q^2}$ ACD code.
	\end{lemma}
	\begin{proof}
		Denote the generator matrices of $\mathscr{C}_{a1}$ and $\mathscr{C}_{a2}$ as $G_{a1}$ and $G_{a2}$, respectively. Then, let $\Phi^{-1}(G_{a1})=(A\mid B)$,  $\Phi^{-1}(G_{a2})=(C\mid D)$,
		and $G_a=(A,C, B, D)$, then there the following equation holds.
		
		$G_a\begin{pmatrix}
			0&I_{n_1+n_2} \\
			-I_{n_1+n_2}&0
		\end{pmatrix}G_a^{T}=\begin{pmatrix}
			-B,-D ,A,C\\
		\end{pmatrix}\begin{pmatrix}
			A^{T}\\
			B^{T}\\
			C^{T}\\
			D^{T}\\
		\end{pmatrix}$
		
		$=(-BA^T+AB^T)+(-DC^T+CD^T)$.
			
			Since, $\mathscr{C}_{a1}$ is ACD code and $\mathscr{C}_{a2}$ is trace Hermitian self-orthogonal code, there are $Rank(-BA^T+AB^T)=2k$ and $-DC^T+CD^T=\mathbf{0}$. Therefore, $\Phi(Ga)$ can generate an ACD code with parameters $(n_1+n_2,k,\ge(d_1+d_2))_{q^2}$.
		\end{proof}
		\begin{example}
			Let $q=2$, $n_1=12$, $n_2=17$. Taking $g_1(x)=x^4 + 1$, $f_{1,0}(x)=x^{11} + x^8 + x^6 + x^5 + x^4 + 1$
			, $f_{1,1}(x)=x^{11} + x^8 + x^4$ and $g_2(x)=x^9 + x^8 + x^6 + x^3 + x + 1$, $f_{2,0}(x)=x^{16} + x^{14} + x^{13} + x^{12} + x^{10} + x^8 + x^7 + x^5 + x^4 + x^3 + x^2 + x$, $f_{2,1}(x)=x^{16} + x^{12} + x^5 + x^4 + x^3 + x^2 + x$. Then, $([g_1(x)f_{1,0}],[g_1(x)f_{1,1}])$ and $([g_2(x)f_{2,0}],[g_2(x)f_{2,1}])$ will generate $[24,8,7]^s_2$ symplectic LCD code $\mathscr{C}_{s1}$ and symplectic self-orthogonal $[34,8,12]^s_2$ code $\mathscr{C}_{s2}$, respectively. Therefore, $\Phi(\mathscr{C}_{s1})$ is $(12,4,7)_4$ ACD code and $\Phi(\mathscr{C}_{s2})$ is additive $(17,4,12)_4$ trace Hermitian self-orthogonal code. In accordance with Lemma \ref{construction_AB}, $\left(\Phi(\mathscr{C}_{s1})\mid \Phi(\mathscr{C}_{s2})\right)$ is an ACD code with parameters $(29,4,\ge 19)_4$, which outperforms $[29,4,18]_4$ LCD code in \cite{ishizuka2022construction}.

			Further, choosing $g_3(x)=x^2 + 1$, $f_{3,0}(x)=x^7 + x^6 + x^5 + x^2$, and $f_{3,1}=x^9 + x^6 + x^5 + x^3 + 1$; then, $([g_3(x)f_{3,0}],[g_3(x)f_{3,1}])$ can also generate an symplectic LCD $[20,8,6]^s_2$ code $\mathscr{C}_{s3}$. By Lemma \ref{construction_AB}, $(\Phi(\mathscr{C}_{s2})\mid \Phi(\mathscr{C}_{s3}))$ is an $(27,4,\ge 18)_4$ ACD code, which is also performing better than LCD code $[27,4,17]_4$ in the literature \cite{ishizuka2022construction}.
		\end{example}

		\begin{lemma}(Construction X of ACD code)\label{ACDconstruction_X}
			If there are two ACD codes $\mathscr{C}_{a2} \subset \mathscr{C}_{a1}$, with parameters $(n,k_2,d_2)_{q^2} \subset (n,k_1,d_1)_{q^2}$, where $d_2 > d_1$. Let $\mathscr{C}_{a3}$ be an additive trace Hermitian self-orthogonal code have parameters $(l,k_1-k_2,\delta )_{q^2}$, then there exists $(n+l,k_1,\min\{\delta+d_1,d_2\} )_{q^2}$ ACD code.
		\end{lemma}
		\begin{proof}
			First, denote the generator matrices of $\mathscr{C}_{a1}$, $\mathscr{C}_{a2}$ and $\mathscr{C}_{a3}$ as $G_{a1}$, $G_{a2}$ and $G_{a3}$, respectively. 
			Since, $\mathscr{C}_{a2} \subset \mathscr{C}_{a1}$, there is $G_{a1}=\begin{pmatrix}
				G_{a2} \\
				G_{ax}
			\end{pmatrix}$.
			Then, construct $Gx=\begin{pmatrix}
				G_{ax}	& G_{a3} \\
				G_{a2}	&0
			\end{pmatrix}$. By Lemma \ref{construction_AB}, $Gx$ can generate an additive code with parameters $(n+l,k_1,\min\{\delta+d_1,d_2\} )_{q^2}$.
		\end{proof}
		\begin{example}
			Let $q=2$, $n=12$. Taking $g(x)=x^4 + 1$, and select $f_1(x)=x^7 + x^6 + x^5 + x^4 + x^3 + x^2 + x + 1$, $f_2(x)=x^9 + x^8 + x^5$, $f_3(x)=x^{10} + x^9 + x^6 + x^5 + x^4 + x^3 + x^2$. Then, $([g(x)],[g(x)f_1(x)],[g(x)f_2(x)],[g(x)f_3(x)])$ will generate a symplectic LCD $[42,8,16]_2^s$ code, which corresponds to an ACD code of parameters $(24,4,16)_4$. Using Magma \cite{bosma1997magma}, one can easily get a subcode with parameter $(24,1,22)_4$ by taking out two codewords of weight $22$ from $(24,4,16)_4$. The generator matrix of additive code $(24,1,22)_4$ is as follows. 
			{\footnotesize 	\setlength{\arraycolsep}{0.01pt}
				$$G_{sub}=\left( {\begin{array}{*{24}{c}}
						1&w&w^2&w&w&0&1&w&w^2&w&w&0&w^2&w^2&1&1&w&w&w^2&w^2&1&1&w&w\\
						w&0&1&w&w^2&w&w&0&1&w&w^2&w&w&w&w^2&w^2&1&1&w&w&w^2&w^2&1&1\\
				\end{array}} \right).$$}
			
			According to Lemma \ref{ACDconstruction_X}, selecting trace Hermitian self-dual $(6,3,4)_4$ \footnote{This code can be derived from quantum $[[6,0,4]]_2$ code in \cite{Grassltable}.} as auxiliary code leads to an ACD code with parameters $(30,4,20)_4$, which has a larger minimum distance than $(30,4,19)_4$ in \cite{ishizuka2022construction}.
		\end{example}

		\begin{lemma}
			If $\mathscr{C}_a$ is an ACD code of parameters $(n,k,d)_{q^2}$, then the following ACD codes also exist:\par
			(1) For $i\ge 1$, $(n+i,k,\ge d)_{q^2}$ (ACD Extend);\par
			(2) For $i \le k$, $(n-i,k-i,\ge d)_{q^2}$ (ACD Shorten);\par
			(3) For $i\le d$, $(n-i,k,\ge d-i)_{q^2}$ (ACD Puncture).
		\end{lemma}
		\begin{proof}
			For (1), it is sufficient to directly juxtapose the all-zero column with $\mathscr{C}_a$, or an additive self-orthogonal code. 
			For (2), let $G_a$ denote generator matrix of $\mathscr{C}_a$. Deleting the first row and column of $G_a$ yields $G_a^{\prime}$. Since the dimensions of ACD codes must be of integer dimensions, the generator matrix of all ACD codes is of even rows. Therefore, $G_a^{\prime}$ will generate an additive code $\mathscr{C}_a^{\prime}$ with $0.5$-dimension hull $\mathcal{H}$. Then, removing $\mathcal{H}$ from $\mathscr{C}_a^{\prime}$ and noting the resulting ACD code as $\mathscr{C}_a^{\prime\prime}$. 
			As the two-step operation from $\mathscr{C}_a$ to $\mathscr{C}_a^{\prime\prime}$ is equivalent to shorten, the parameter of $\mathscr{C}_a^{\prime\prime}$ is $(n-1,k-1,\ge d)_{q^2}$. Repeating this process will yield ACD codes with parameters $(n-i,k-i,\ge d)_{q^2}$.
			Finally, since shorten for $\mathscr{C}_a$ is equivalent to puncture for $\mathscr{C}_a^{\perp}$, (3) also holds.
		\end{proof}

		With a computer-aided search method, we also construct some good ACD codes of lengths ranging from $22$ to $30$. We give their generators in Table \ref{good_QCACD} in Appendix. Further, we compare them with the best quaternary LCD codes in \cite{lu2020optimal,harada2021construction, Ishizuka2022,ishizuka2022construction} in Table \ref{good_ACD}. The results show that our ACD codes have greater distances for the same length and dimension.
		Specifically, in Table \ref{good_ACD}, ``S" represents the shorten code.
		
		
		\begin{table}[htbp]
			\caption{Quaternary ACD codes outperform linear counterparts}\label{good_ACD}
			\centering
			\begin{tabular}{ccc}
				\toprule 
				No. &   Our ACD Codes    & Best LCD Codes in \cite{lu2020optimal,harada2021construction,Ishizuka2022,ishizuka2022construction} \\ \midrule
				1  &   $(22,10,9)_4$    &                                            $[22,10,8]_4$                                             \\
				2  &   $(25,13,8)_4$    &                                            $[25,13,7]_4$                                            \\
				3  &   $(26,6,15)_4$    &                                            $[26,6,14]_4$                                            \\
				4  &   $(26,14,8)_4$    &                                            $[26,14,7]_4$                                            \\
				5  &   $(27,10,12)_4$   &                                           $[27,10,11]_4$                                            \\
				6  &  $(26,9,12)_4$(S)  &                                            $[26,9,10]_4$                                            \\
				7  &   $(27,12,10)_4$   &                                            $[27,12,9]_4$                                            \\
				8  & $(26,11,10)_4$(S)  &                                            $[26,11,9]_4$                                            \\
				9  &   $(27,15,8)_4$    &                                            $[27,15,7]_4$                                            \\
				10  &   $(28,12,11)_4$   &                                            $[28,12,9]_4$                                             \\
				11  &   $(29,14,10)_4$   &                                            $[29,14,9]_4$                                            \\
				12  & $(28,13,10)_4$(S)  &                                            $[28,13,9]_4$                                            \\
				13  &   $(27,4,18)_4$    &                                            $[27,4,17]_4$                                             \\
				14  &   $(29,4,19)_4$    &                                            $[29,4,18]_4$                                            \\
				15  &   $(30,14,10)_4$   &                                            $[30,14,9]_4$                                            \\
				16  & $(29,13,10)_4$(S)  &                                            $[29,13,9]_4$                                            \\
				17  & $(30,4,20)_4$  &                                            $[30,4,19]_4$                                            \\ \bottomrule
			\end{tabular}%
		\end{table}
		
		\section{Conclusion}\label{VI}
		In this work, we propose a lower bound on the symplectic distance of $1$-generator quasi-cyclic codes with index even and give several combinations and derivations of additive codes. To verify the applicability of our methods, we construct many good additive codes and ACD codes that are better than the best-known linear codes in \cite{Grassltable} and the best-known LCD codes in \cite{lu2020optimal,harada2021construction, Ishizuka2022,ishizuka2022construction}, respectively.
		Our results show that in the optimal case, additive codes can have higher code rates than linear codes; in the non-optimal case, additive codes can have larger distances than linear codes of the same lengths and dimensions. However, it remains an open problem whether the distance of optimal additive codes can be greater than that of optimal linear counterparts.
		In addition, it is notable that most of the additive codes in this paper can also be considered additive cyclic codes. Therefore, it will be an interesting problem to study the construction of additive cyclic codes in the future.
		
		\section*{Appendix}
		
		In order to save space, we give generators of quasi-cyclic codes in abbreviated form in Table \ref{good_QC} and \ref{good_QCACD}, presenting the coefficient polynomials in ascending order, with the indexes of the elements representing successive elements of the same number.
		For example, polynomial $1 + x^2 + x^3 + x^4$ over $\mathbb{F}_2$ is denoted as $101^3$. 
		Some of the additive codes in Tables \ref{good_QC} and \ref{good_QCACD} are derived codes and are marked with abbreviations to save space, as follows.
		\begin{itemize}
			\item D: Dual Code;	
			\item Au: Augment Code (Add $\mathbf{1_n}$);	
			\item DoubleAu: Augment Twice Code (Add $\mathbf{1_n}$ and $\mathbf{w_n}$);
			\item X: Additive Construction X.
		\end{itemize}
		
		\begin{table*}[h!]
			\centering
			\caption{Quaternary additive codes from symplectic 1-generator quasi-cyclic codes}\label{good_QC}
			
			\resizebox{\textwidth}{!}{
				\begin{threeparttable}

					\begin{tabular}{cccc}
						\toprule
						No. &  $\mathscr{C}_a$  &                                                                                                                                                                                                                                                                                                                                                                                                                                               $g(x),f_0(x),f_1(x)$                                                                                                                                                                                                                                                                                                                                                                                                                                               &                                  Derived Codes                                   \\ \midrule
						1  &   $(21,10,8)_4$   &                                                                                                                                                                                                                                                                                                                                                                                                           \makecell[c]{$1^{2}$,$01010^{2}1010^{2}1^{2}0^{2}1^{2}0^{3}1$,$1^{3}0^{5}1^{6}0^{2}1^{3}01$}                                                                                                                                                                                                                                                                                                                                                                                                           &                               $(21,10.5,8)_4$ (Au)                               \\ \midrule
						2  &  $(21,8.5,9)_4$   &                                                                                                                                                                                                                                                                                                                                                                                                         \makecell[c]{$1^{3}01$,$01^{2}0^{4}1^{5}0^{2}10^{2}1^{2}01$, $01^{2}01^{4}0^{2}1^{4}0^{3}1^{4}$}                                                                                                                                                                                                                                                                                                                                                                                                         &                               $(22,12.5,7)_4$(ExD)                               \\ \midrule
						3  &  $(21,5.5,12)_4$  &                                                                                                                                                                                                                                                                                                                                                                                                               \makecell[c]{$10^{2}1^{2}010101$,$1^{5}010^{2}1^{4}01^{4}01^{2}$,$0^{2}101^{7}0101$}                                                                                                                                                                                                                                                                                                                                                                                                               &                               $(22,15.5,5)_4$(ExD)                               \\ \midrule
						4  &  $(30,6.5,18)_4$  &                                                                                                                                                                                                                                                                                                                                                                                          \makecell[c]{$10101^{4}0101^{2}0^{4}1$,$10^{2}101^{2}0^{2}1^{7}0^{3}1^{5}01^{3}$,$0^{2}101^{2}0^{2}10^{3}10^{5}1^{4}01^{5}01$}                                                                                                                                                                                                                                                                                                                                                                                          &                                        -                                         \\ \midrule
						5  &  $(35,9.5,18)_4$  &                                                                                                                                                                                                                                                                                                                                                                                       \makecell[c]{$101010^{2}1^{6}0^{3}1$,$101^{2}0^{3}1^{3}01^{3}01^{3}01^{9}01^{2}0^{2}1^{2}$,\\$1010^{13}101^{2}010^{2}1^{4}01^{2}0101$}                                                                                                                                                                                                                                                                                                                                                                                       &                                        -                                         \\ \midrule
						6  & $(42,17.5,14)_4$  &                                                                                                                                                                                                                                                                                                                                                                   \makecell[c]{$10^{2}101^{3}$, $101^{2}01^{2}01^{2}01^{2}0^{7}10^{2}10^{2}10^{2}101^{3}01010^{2}1$,\\ $010101^{2}01^{3}0101^{5}0^{2}1^{2}010^{2}1^{3}0^{2}10^{4}10^{2}1^{2}$}                                                                                                                                                                                                                                                                                                                                                                   &                               $(42,24.5,10)_4$(D)                                \\ \midrule
						7  & $(45,14.5,19)_4$  &                                                                                                                                                                                                                                                                                                                                                         \makecell[c]{$1^{3}0^{2}1^{3}0^{4}1^{5}$, $0^{2}10^{3}1^{4}010^{2}1^{2}0101^{2}0^{2}10^{2}1^{2}01^{2}01^{2}0^{2}101^{2}0^{3}1^{2}$,\\ $0101^{2}01^{5}0^{3}101010101^{2}0^{7}1^{2}0^{4}101^{6}$}                                                                                                                                                                                                                                                                                                                                                          &                                        -                                         \\ \midrule
						8  & $(45,15.5,18)_4$  &                                                                                                                                                                                                                                                                                                                                                                        \makecell[c]{$1^{6}0^{6}1^{3}$, $0^{2}1^{7}0^{3}1010^{2}10^{6}1^{7}0^{2}1^{4}0^{3}1^{4}$,\\$0^{5}1010^{4}1^{2}0^{2}101010^{2}1^{4}01^{2}010^{2}1^{2}0^{2}1010^{2}1$}                                                                                                                                                                                                                                                                                                                                                                        &                                        -                                         \\ \midrule
						9  & $(45,16.5,17)_4$  &                                                                                                                                                                                                                                                                                                                                                         \makecell[c]{$10^{8}10^{2}1$, $101^{2}01^{2}0101010^{2}1^{3}01^{2}01^{2}0^{2}1^{4}0^{2}10^{3}10^{3}101^{3}$,\\$1^{3}0^{2}101^{2}010^{2}1010^{2}1^{2}0^{2}101^{2}01010101^{3}0^{2}1^{2}0101^{2}$}                                                                                                                                                                                                                                                                                                                                                         &                                        -                                         \\ \midrule
						10  & $(47,11.5,24)_4$  &                                                                                                                                                                                                                                                                                                                                        \makecell[c]{$10^{3}1^{2}0^{2}1^{2}01^{2}0^{2}10^{2}1010^{2}1^{2}$,\\$01^{2}01^{3}0^{3}1010^{4}101^{2}0^{3}10^{4}1^{2}0101^{2}01^{3}01^{2}0101$,\\ $1^{3}0^{4}1^{2}0^{4}1010^{3}1^{3}0^{2}1^{2}0^{5}1^{2}0^{2}101^{2}010^{2}1^{2}$}                                                                                                                                                                                                                                                                                                                                         &                                        -                                         \\ \midrule
						11  & $(71,17.5,32)_4$  &                                                                                                                                                                                                                                                                                                               \makecell[c]{$1010101^{2}010^{3}1^{2}0^{2}1^{2}0^{5}10^{4}10^{3}1^{4}$,\\$0^{2}1^{6}0^{4}1^{12}010^{2}101010101^{4}010^{2}1^{4}01^{2}0^{3}101^{3}010^{4}101^{2}01$,\\ $101^{3}01^{3}0^{3}1^{2}0^{2}1^{3}01^{4}0^{5}1^{2}010^{6}1^{2}01^{2}01^{4}010^{2}10^{2}1^{3}0^{2}1^{7}0^{2}1$}                                                                                                                                                                                                                                                                                                               &                                        -                                         \\ \midrule
						12  & $(73,13.5,38)_4$  &                                                                                                                                                                                                                                                                                \makecell[c]{$1^{2}0^{2}1^{6}0^{4}1^{3}0^{2}101^{3}0^{2}10^{3}1010101^{2}0^{3}10^{3}1^{2}$,\\ $01^{2}0^{2}1^{2}010^{2}1010^{4}101^{2}0^{3}1^{3}0^{2}1^{2}01^{3}0^{2}10^{3}10^{2}10^{3}1010^{3}1^{3}0^{3}1^{3}01010^{2}1$,\\$1^{4}0101^{4}0^{3}1^{3}0^{4}1010^{2}1^{3}0^{2}101^{3}010^{2}101^{3}01^{2}01^{3}0^{2}101^{3}0^{3}1^{2}01^{3}0^{2}101$}                                                                                                                                                                                                                                                                                 &                                        -                                         \\ \midrule
						13  &   $(91,6.5,63)_4$   &                                                                                                                                                                                                                      \makecell[c]{$1^{2}010^{3}1^{2}01^{2}0^{2}10^{4}1010101^{2}0101^{3}0^{2}10^{2}101^{2}01^{2}0^{2}1^{3}0^{5}10^{2}1^{2}0^{3}1^{2}0^{2}1^{4}01010101$,\\ $1^{2}0^{2}1^{2}01^{3}0^{3}1^{2}01^{2}0^{3}1^{3}010^{3}101^{2}0^{2}1^{4}010^{2}10^{2}10101^{6}0^{2}1^{2}0^{2}101^{2}01^{5}01010^{3}1^{2}0^{2}1^{3}0^{3}1^{2}$,\\ $01^{2}0^{3}1^{2}0^{3}10^{4}101^{2}0^{3}1^{3}0^{2}1^{3}01^{3}0^{2}10^{2}1^{2}010^{3}10^{3}1010^{2}1010^{4}1^{3}01^{3}01^{3}01^{3}0^{9}1^{2}01$}                                                                                                                                                                                                                      &              \makecell[c]{$(91,7,63)_4$(Au) \\$(92,7,64)_4$(ExAu)}               \\ \midrule
						14  & $(124,4.5,92)_4$  &                                                                                                                                               \makecell[c]{$101^{5}01^{4}0^{2}101^{2}01^{6}0^{2}1^{2}0^{3}1^{3}0^{2}10^{3}1^{5}0^{4}10^{2}1^{2}0101010^{4}1010^{3}101^{2}0^{4}$\\$1^{3}01010^{2}1^{2}0^{2}10^{2}10^{2}1^{3}01^{2}0101^{2}010^{3}1^{2}$,$01^{2}0101^{3}01^{3}01^{8}01^{4}0^{2}1^{3}0101^{2}0^{5}$\\$1^{2}0^{2}1^{2}01^{2}0^{3}10101010^{4}10^{2}10^{3}10^{2}1^{5}0101^{3}010^{3}1^{2}010^{4}10^{3}1^{6}0^{3}1^{3}010^{2}1^{5}$,\\$1^{3}010101^{2}01^{2}0101^{2}01^{4}0^{6}101^{2}01^{2}010^{7}1^{2}0101^{5}0^{2}10^{2}1^{3}01^{2}$\\$0^{2}10^{2}1^{2}0^{2}1^{4}0^{2}1^{2}01^{2}0^{2}1^{4}0101^{3}0^{3}10^{3}1010101^{2}0^{2}1^{2}0^{2}1^{2}01^{2}$}                                                                                                                                               &                               -                               \\ \midrule
						15  & $(155,4.5,116)_4$\tnote{4} &                                                                          \makecell[c]{$1^{3}010^{3}1^{2}010101^{2}010^{2}1^{2}0^{2}1^{2}0^{3}101^{2}0^{2}10^{2}101^{2}01^{2}0101^{5}01010^{2}10^{2}1^{5}0^{3}$\\$1^{3}0^{4}10^{3}1^{4}0^{2}1010^{3}10^{2}
							1^{2}01^{7}0^{2}1^{3}01^{2}0^{4}1^{2}0^{5}101^{3}01^{3}0^{2}10^{4}10101$,\\ $01^{2}010^{6}10^{3}1010^{2}1^{3}0^{2}101^{2}0^{4}1^{2}01^{7}0^{2}101^{3}0^{4}10^{10}10^{3}1^{4}0^{5}1^{2}0^{2}1^{2}0^{3}1^{2}$\\$0^{3}1^{2}0^{2}1^{3}0^{2}1^{3}0^{3}10^{2}10^{5}101^{6}0101^{7}010^{2}1^{2}0^{3}1^{3}010^{4}1^{3}$,\\ $01^{2}0^{2}10^{5}10^{3}1010101^{4}01010101^{3}01^{7}
							0^{2}1^{2}0101^{2}0^{4}1^{3}0^{2}10^{2}1010^{3}1^{2}01^{2}01^{3}0$\\$10^{2}10^{3}1^{3}0^{2}1^{2}010^{2}1^{2}010^{2}10^{3}1010^{2}10101^{2}010^{2}10^{2}1^{2}0^{2}1010^{2}10^{3}1^{3}010^{3}1$}                                                                          &                               $(166,4.5,124)_4$(X)                              \\ \midrule
						16  & $(195,6.5,139)_4$ & \makecell[c]{$10^{2}101^{2}010^{3}101^{2}0^{2}1^{2}01^{2}0^{3}10101^{5}0^{3}1^{3}010^{5}1^{3}0101^{3}0^{2}10^{7}10^{7}1^{3}01^{3}0^{5}10^{2}101$\\$01^{7}0^{6}1^{2}0^{2}10^{2}1^{5}0^{3}1^{2}01^{5}010^{2}10^{2}10
							10^{2}1^{4}010^{2}101^{3}0^{5}1^{5}01^{2}0^{3}1^{2}01^{2}0101^{2}01$,\\ $0^{4}1^{5}0^{2}1^{2}0^{7}1^{3}0^{2}1^{3}0^{2}1010^{2}101^{2}0^{4}1^{2}0^{2}1010^{4}1^{2}0^{2}1^{3}01^{2}0^{2}10^{3}10^{3}10^{3}1^{2}0^{4}1^{2}0$\\$101^{3}010^{2}10^{3}10^{2}101^{2}0^{2}1^{2}01^{11}010^{8}1^{2}0^{4}10^{2}101^{2}0^{2}1^{2}0^{5}101^{8}010^{4}1^{4}01^{2}01^{4}0^{2}1^{2}0^{2}1010^{2}1$,\\ $0^{2}1^{3}0^{2}10^{3}101^{3}01010^{2}1^{3}0^{4}1^{3}01010^{2}10^{2}10^{5}10^{3}1^{2}0^{2}10^{6}1^{2}0^{4}101^{3}010^{2}1^{3}0^{3}1^{3}0^{4}1^{4}$\\$0^{3}1^{2}010^{2}10^{4}1^{4}01^{4}0^{5}101^{4}01^{3}01
							0^{2}1^{2}01^{2}010101010^{4}101^{4}01^{2}0101^{3}0^{2}1^{4}0^{7}1^{2}0101^{6}0^{2}1$} & \makecell[c]{$(196,6.5,140)_4$(Ex) \\$(195,7,139)_4$(Au)\\$(196,7,140)_4$(ExAu)} \\ \bottomrule
					\end{tabular}
					\begin{tablenotes}    
						\footnotesize               
						\item[4]  By taking out the code with the weight 124 from $(155,4.5,116)_4$, an additive code with parameter $(155,2,124)_4$ can be generated. Therefore, with $(155,2,124)_4 \subset (155,4.5,116)_4$, we can obtain $(166,4.5,124)_4$ by choosing $(11,2.5,8)_4$ as auxiliary code.        
					\end{tablenotes}            
			\end{threeparttable}  }     
		\end{table*}

		\begin{table*}[h!]
			\caption{Quaternary ACD codes from symplectic 1-generator quasi-cyclic codes}\label{good_QCACD}
			\centering
			\begin{tabular}{cccc}
				\toprule
				No. & $\mathscr{C}_a$ &                                           $g(x),f_0(x),f_1(x)$                                            &  Derived Codes   \\ \midrule
				1     & $(22,10,9)_4$ & $101$, $1010^{7}1^{2}0^{2}10^{4}1$, $10^{4}10^{3}1^{2}0^{3}101^{4}01$ & - \\
				2     & $(25,12,9)_4$ & $1^2$, $01^{2}01010^{4}1010^{3}1^{2}0^{3}1^3$, $1^{4}010^{3}1^{4}010^{2}10^{2}10^{1}101$ & $(25,13,8)_4$(D) \\
				3     & $(26,12,8)_4$ & $101$, $10^{3}1^{3}0^{3}1^{3}010^{2}10101^{2}01^2$, $10^{2}1^{3}010101^{2}0^{4}1^{2}0^{5}1^2$ & $(26,14,8)_4$(D) \\
				4     & $(27,10,12)_4$ & $1^{2}01^{2}01^2$, $0^{4}10^{3}1^{3}01^{3}01^{2}0^{7}1^2$, $0^{2}10^{2}1010^{2}1^{2}01^{2}0^{6}1010^{2}1$ & - \\
				5     & $(27,12,10)_4$ & $10^{2}1$, $0^{4}101^{2}0^{4}1^{5}010101^{3}01$, $0^{6}1^{2}010101^{2}010^{5}10^{2}1$ & - \\
				6     & $(27,15,8)_4$ & $10^{2}1$, $1$, $0101^{5}0^{3}1^{2}01$ & - \\
				7     & $(28,12,11)_4$ & $10^{3}1$, $1^{6}01^{2}0^{3}1^{2}0^{2}10^{2}1^{2}01^{3}0^{2}1$, $0^{2}1^{2}01^{5}0101010^{8}101$ & - \\
				8     & $(29,14,10)_4$ & $1^2$, $1010^{5}10^{2}1^{3}0^{2}10^{3}10^{2}101$, $10^{5}1^{3}01^{3}01^{2}01^{6}010^{3}1$ & - \\
				9     & $(30,14,10)_4$ & $101$, $1^{2}0^{4}10^{3}101^{2}0^{3}101^{5}0^{3}101$, $1^{3}0^{2}101^{2}010^{2}1^{6}0101^{2}0^{5}1$ & - \\
				\bottomrule
			\end{tabular}%
		\end{table*}

		%

		\bibliographystyle{IEEEtran} 
		\bibliography{reference} 

\end{document}